\documentclass[12pt]{article}
\usepackage{amsmath, amsfonts, amssymb, amsthm, graphicx, geometry, arydshln, umoline, subfig, color, comment, enumerate, bm,float,multirow,booktabs,varwidth}
\IfFileExists{newtxtext.sty}{\usepackage{newtxtext,newtxmath}}{}
\usepackage{tikz}
\usetikzlibrary{matrix, positioning,fit,calc,arrows.meta}
\usepackage{appendix}
\usepackage[colorlinks=true,citecolor=blue]{hyperref}
\usepackage{natbib}
\usepackage{booktabs}
\usepackage[table]{xcolor}
\def\BibTeX{{\rm B\kern-.05em{\sc i\kern-.025em b}\kern-.08em
    T\kern-.1667em\lower.7ex\hbox{E}\kern-.125emX}}
\usepackage{makecell}

\renewcommand{\baselinestretch}{1.5}
\numberwithin{equation}{section}

\theoremstyle{definition}
\newtheorem{theorem}{Theorem}

\newcommand{\mR}{\mathbb{R}}

\newcolumntype{L}[1]{>{\raggedright\arraybackslash}p{#1}}

\newcommand{\cA}{\mathcal{A}}
\newcommand{\cL}{\mathcal{L}}
\newcommand{\cN}{\mathcal{N}}

\newcommand{\profiletable}[4]{%
\begin{tikzpicture}[baseline=(P.center)]
\matrix (P) [
matrix of nodes,
ampersand replacement=\&,
nodes={anchor=center,inner xsep=5pt,inner ysep=4pt},
row sep=0pt,
column sep=0pt,
column 1/.style={nodes={text width=#1,align=center}},
column 2/.style={nodes={text width=#2,align=center}},
row 1/.style={nodes={font=\scriptsize\bfseries,minimum height=1.25cm}}
]{
Voters \& \makecell{Order from highest\\to lowest} \\
#3 \& #4 \\
};
\draw (P-1-1.south west) -- (P-1-2.south east);
\draw (P-1-1.north east) -- (P-2-1.south east);
\end{tikzpicture}%
}

\title{Characterizing the Plurality Rule via Clone Invariance}
\author{Leo Goto\thanks{Undergraduate School of Management, Department of Business Economics, Tokyo University of Science, 1-11-2, Fujimi, Chiyoda-ku, Tokyo, 102-0071, Japan. Email: leogotodeb@gmail.com.}}
\date{August 19, 2026}

\begin{document}
\maketitle

\begin{abstract}
The Plurality rule is vulnerable to vote splitting between similar alternatives.
Defying conventional wisdom, we characterize the Plurality rule through a specific type of clone independence condition.
We further introduce the notion of a one-parameter family of \textit{$p$-Clone Resistance} (\textit{$p$-CR}), and examine the property within the class of scoring rules.
For $p \,\in (0,1)$, nontrivial rules exist with at most three alternatives, but only the Trivial rule generally satisfies the condition.


\end{abstract}

\section{Introduction}
An election can have a different winner simply because a losing alternative enters.
Under the Plurality rule, two similar alternatives may divide the same pool of first-place support, allowing a rival to win.
This is called the \textit{spoiler effect}, and it raises a basic question: which voting rules are immune to the entry of a similar alternative, and how does the answer depend on the way voters rank the similar pair?
This question motivates \textit{Independence of Clones}, introduced by \citet{tideman1987} to formalize this vulnerability of the Plurality rule, and the recent analysis of approximate clones by \citet{delemazure2026}.

Consider a situation where a similar alternative $x'$ of $x$ is available.
By adding $x'$, we suppose that $x$ and $x'$ take adjacent positions in every ranking.
Among voters with the same ranking over the original alternatives, a fixed fraction $p$ places $x$ above $x'$.
\textit{$p$-Clone Resistance} (\textit{$p$-CR}) preserves choices among the original alternatives, and \textit{$p$-Clone Invariance} (\textit{$p$-CI}) requires the entire choice set to remain unchanged.
At $p \,=\, 1$, \textit{$p$-Clone Invariance}, together with \textit{Anonymity}, \textit{Neutrality}, and \textit{Consistency}, characterizes the Plurality rule.

We then examine positional scoring rules satisfying \textit{$p$-CR}.
With at most three alternatives, nontrivial rules exist for every $p$. Their two-alternative restriction is Plurality for $p \,> 1/2$ and inverse Plurality for $p \,< 1/2$, with either at $p \,=\, 1/2$.
For $p \,\in (0,1)$, no nontrivial rule extends to four alternatives.
At the endpoints, Plurality rule satisfies the condition for $p \,=\, 1$, and inverse Plurality rule for $p \,=\, 0$.

The Plurality rule has also been characterized by \textit{Independence of dominated alternatives} \citep{richelson1978,ching1996}, \textit{Tops-only} conditions with \textit{Efficiency} or \textit{Faithfulness} \citep{yeh2008,sekiguchi2012}, and \textit{Fixed-population monotonicity} \citep{kellyqi2016}.
In a $1$-clone extension, $x'$ is ranked immediately below $x$ by every voter.
It is thus unanimously dominated by $x$ and inherits $x$'s comparison with every other alternative.
\textit{Independence of dominated alternatives} therefore implies \textit{$1$-CI}, and Theorem~\ref{thm:plurality} recovers the characterization in \citet{ching1996} from a formally weaker requirement.
Other research on adjacent clones consider ranked pairs \citep{zavisttideman1989}, composition and cloning consistency \citep{laffondetal1996,laslier2000,ozturk2020}, manipulation \citep{elkindetal2011}, and approximate clones \citep{delemazure2026}.

\section{Model}

Let $\cN$ and $\cA$ be infinite sets of potential voters and alternatives. Let $N \,\subset \cN$ and $A \,\subset \cA$ be finite and nonempty. A strict linear order $\succ_i$ on $A$ is asymmetric, transitive, and complete on distinct alternatives. Let $\cL(A)$ denote the set of all such orders. A preference profile is $\succ \,=\, (\succ_i)_{i \,\in N} \,\in \cL(A)^N$.

For $B \,\subseteq A$, the restriction of $\succ$ to $B$ is
\begin{equation}\label{eq:restriction}
\left.\succ\right|_B \,=\, (\succ_i^B)_{i \,\in N},
\qquad
a\succ_i^B b\quad \Longleftrightarrow\quad a\succ_i b
\quad\text{for all }a,b \,\in B.
\end{equation}
The rank of $a \,\in A$ in voter $i$'s order is
\begin{equation}\label{eq:rank}
r_i(a,\succ) \,=\, 1+\bigl|\{b \,\in A\mid b\succ_i a\}\bigr|.
\end{equation}
A social choice rule $F$ assigns a nonempty set $F(\succ) \,\subseteq A$ to every profile for every finite $N$ and $A$.

We also define the transformations used below. For a bijection $\pi:N\to N'$, the relabelled profile $\pi\succ$ satisfies $(\pi\succ)_{\pi(i)} \,=\, \succ_i$. For a permutation $\sigma:A\to A$, the profile $\sigma\succ$ satisfies
\[
\sigma(a)\mathrel{(\sigma\succ_i)}\sigma(b)
\quad\Longleftrightarrow\quad
a\succ_i b.
\]
For profiles $\succ \,\in \cL(A)^N$ and $\succ' \,\in \cL(A)^{N'}$ with $N\cap N' \,=\, \emptyset$, $\succ\cup\succ'$ denotes the profile on $N\cup N'$ that agrees with each constituent profile on its electorate.

Fix $p \,\in \mathbb Q\cap[0,1]$, $x \,\in A$, and $x' \,\notin A$. For $\rho \,\in \cL(A)$, let $n_\rho(\succ) \,=\, |\{i \,\in N\mid\succ_i \,=\, \rho\}|$. A profile $\widetilde\succ \,\in \cL(A\cup\{x'\})^N$ is a \textit{$p$-clone extension} of $x$ at $\succ$ if
\begin{align}
\left.\widetilde\succ\right|_A&\,=\,\succ,\label{eq:clone-restriction}\\
\bigl|r_i(x,\widetilde\succ)-r_i(x',\widetilde\succ)\bigr|&\,=\,1
\quad\text{for every }i \,\in N,\label{eq:clone-adjacency}\\
\bigl|\{i \,\in N\mid \succ_i \,=\, \rho\text{ and }x\mathrel{\widetilde\succ_i}x'\}\bigr|
&\,=\,p\,n_\rho(\succ)
\quad\text{for every }\rho \,\in \cL(A).\label{eq:clone-share}
\end{align}
Thus, the extension preserves all original comparisons, makes $x$ and $x'$ adjacent, and fixes their relative order among voters having the same original order $\rho$. Such an extension exists only when every $p\,n_\rho(\succ)$ is an integer. We multiply all voter counts by a common positive integer when necessary. Figure~\ref{fig:binary-extension} illustrates the construction.

\begin{figure}[H]
\centering
\begin{minipage}[c]{0.4\textwidth}
\centering
\small Original profile\par\medskip
\profiletable{1.35cm}{3.00cm}
{\(\begin{array}{c}u\\v\end{array}\)}
{\(\begin{array}{cc}\colorbox{blue!14}{$x$}&y\\y&\colorbox{blue!14}{$x$}\end{array}\)}
\end{minipage}
\hfill
\raisebox{0.15cm}{\tikz\draw[-{Stealth[length=2.2mm]},thick] (0,0)--(0.75,0);}
\hfill
\begin{minipage}[c]{0.47\textwidth}
\centering
\small $p$-clone extension\par\medskip
\profiletable{1.85cm}{3.15cm}
{\(\begin{array}{c}pu\\(1-p)u\\pv\\(1-p)v\end{array}\)}
{\(\begin{array}{ccc}
\colorbox{blue!14}{$x$}&\colorbox{orange!18}{$x'$}&y\\
\colorbox{orange!18}{$x'$}&\colorbox{blue!14}{$x$}&y\\
y&\colorbox{blue!14}{$x$}&\colorbox{orange!18}{$x'$}\\
y&\colorbox{orange!18}{$x'$}&\colorbox{blue!14}{$x$}
\end{array}\)}
\end{minipage}
\caption{A two-alternative profile and its $p$-clone extension.}
\label{fig:binary-extension}
\end{figure}

\subsection{Axioms and scoring rules}

\begin{itemize}
\item[] \textbf{\textit{Anonymity}.} For every profile $\succ$ and every bijection $\pi:N\to N'$, $F(\pi\succ) \,=\, F(\succ)$.
\item[] \textbf{\textit{Neutrality}.} For every profile $\succ$ and every permutation $\sigma:A\to A$, $F(\sigma\succ) \,=\, \sigma(F(\succ))$.
\item[] \textbf{\textit{Consistency}.} For profiles $\succ$ and $\succ'$ on the same alternatives and disjoint electorates, if $F(\succ)\cap F(\succ') \,\ne\, \emptyset$, then
\[
F(\succ\cup\succ') \,=\, F(\succ)\cap F(\succ').
\]
\item[] \textbf{\textit{$p$-Clone Resistance} (\textit{$p$-CR}).}
For every profile $\succ$, every $x \,\in A$, and every $p$-clone extension $\widetilde\succ$ of $x$ at $\succ$,
\begin{equation}\label{eq:cei}
F(\widetilde\succ)\cap A \,=\, F(\succ).
\end{equation}
\item[] \textbf{\textit{$p$-Clone Invariance} (\textit{$p$-CI}).}
For every profile $\succ$, every $x \,\in A$, and every $p$-clone extension $\widetilde\succ$ of $x$ at $\succ$,
\begin{equation}\label{eq:pic}
F(\widetilde\succ) \,=\, F(\succ).
\end{equation}
\end{itemize}

The last condition strengthens \textit{$p$-CR} by excluding the entrant.

A positional scoring rule $F^s$ has a vector $s^m \,=\, (s_1^m,\ldots,s_m^m) \,\in \mR^m$ for each number $m$ of alternatives. The total score of $a \,\in A$ is
\begin{equation}\label{eq:score}
S_a(\succ) \,=\, \sum_{i \,\in N}s^{|A|}_{r_i(a,\succ)},
\end{equation}
and the rule selects
\begin{equation}\label{eq:scoring-rule}
F^s(\succ) \,=\, \{a \,\in A\mid S_a(\succ) \,\ge\, S_b(\succ)\text{ for every }b \,\in A\}.
\end{equation}
Two vectors of the same dimension are equivalent if one is obtained from the other by adding a common constant and multiplying by a positive number. Equivalent vectors induce the same rule. The Trivial rule uses constant vectors. A nontrivial rule is any rule other than the Trivial rule.
Plurality rule, inverse Plurality rule, Borda, and inverse Borda use $(1,0,\ldots,0)$, $(0,\ldots,0,1)$, $(m-1,m-2,\ldots,0)$, and $(0,1,\ldots,m-1)$, respectively.

A simple scoring rule uses one score vector. A composite scoring rule uses a finite ordered list and applies each vector only to alternatives tied under every preceding vector. We call the vectors in this list its component vectors.

\section{Results}

The following theorem characterizes the Plurality rule.
\begin{theorem}\label{thm:plurality}
Let $\bar m \,\ge\, 3$. A social choice rule defined for every finite electorate and every set containing at most $\bar m$ alternatives satisfies \textit{Anonymity}, \textit{Neutrality}, \textit{Consistency}, and \textit{$1$-CI} if and only if it is the Plurality rule.
\end{theorem}

The proof uses a complete classification of positional scoring rules shown in Theorem~\ref{thm:cei}.
For $m \,\ge\, 2$, define
\[
d_p^m \,=\, (1,\underbrace{1-p,\ldots,1-p}_{m-2},0),\qquad
e_p^m \,=\, (0,\underbrace{p,\ldots,p}_{m-2},1).
\]

The following result categorizes all positional scoring rules satisfying \textit{$p$-CR} for each $p$.
\begin{theorem}\label{thm:cei}
Let $p \,\in \mathbb Q\cap[0,1]$, and let a positional scoring rule be defined for every number of alternatives up to $\bar m \,\ge\, 3$.
\begin{enumerate}[(i)]
\item If $\bar m \,=\, 3$, the rule satisfies \textit{$p$-CR} if and only if it is Trivial, uses $(1,0)$ and $d_p^3$ when $p \,>\, 1/2$, uses either $(1,0)$ and $d_{1/2}^3$ or $(0,1)$ and $e_{1/2}^3$ when $p \,=\, 1/2$, or uses $(0,1)$ and $e_p^3$ when $p \,<\, 1/2$.
\item If $\bar m \,\ge\, 4$ and $p \,\in (0,1)$, only the Trivial rule satisfies \textit{$p$-CR}.
\item If $\bar m \,\ge\, 3$, the rules satisfying \textit{$1$-CR} are the Plurality rule and the Trivial rule, while those satisfying \textit{$0$-CR} are the inverse Plurality rule and the Trivial rule.
\end{enumerate}
\end{theorem}

\begin{proof}
For $m \,=\, 2$, every vector is constant or equivalent to $(1,0)$ or $(0,1)$. An extension from one to two alternatives preserves $x$ under $(1,0)$ exactly when $p \,\ge\, 1/2$, and under $(0,1)$ exactly when $p \,\le\, 1/2$. Now fix $m \,=\, k+1 \,\ge\, 3$ and take original alternatives $a_1,\ldots,a_k$. Choose a positive integer $q$ such that $pq$ is an integer. For each $j \,=\, 1,\ldots,k$, include $q$ voters with order
\[
a_j\succ a_{j+1}\succ\cdots\succ a_{j+k-1},
\]
where $a_{k+\ell} \,=\, a_\ell$. Each alternative then appears $q$ times at every rank, so all alternatives tie. Add a clone $a_1'$ of $a_1$ and write the new $m$-dimensional vector as $(a,b_2,\ldots,b_k,c)$. All original alternatives must remain chosen. Comparing $a_j$ with $a_{j+1}$ for $j \,=\, 2,\ldots,k-1$, and then comparing $a_1$ with $a_2$, gives
\begin{equation}\label{eq:shape}
b_2 \,=\, \cdots \,=\, b_k \,=\, b,\qquad b \,=\, (1-p)a+pc.
\end{equation}
The score difference between $a_1$ and $a_1'$ is $q(2p-1)(a-c)$. Since $a_1$ must remain chosen, $a \,\ge\, c$ when $p \,>\, 1/2$ and $a \,\le\, c$ when $p \,<\, 1/2$. Equation~\eqref{eq:shape} therefore gives a constant vector or $d_p^m$ in the first case, and a constant vector or $e_p^m$ in the second. When $p \,=\, 1/2$, it gives a constant vector, $d_{1/2}^m$, or $e_{1/2}^m$.

For three alternatives and $p \,\ge\, 1/2$, use $(1,0)$ before entry and $d_p^3$ after entry. In Figure~\ref{fig:binary-extension}, $u$ voters rank $x$ above $y$ and $v$ voters rank $y$ above $x$. The resulting score differences are
\[
S_x-S_y \,=\, (1-p+p^2)(u-v),\qquad
S_x-S_{x'} \,=\, (2p-1)(pu+(1-p)v).
\]
Since $1-p+p^2 \,>\, 0$ and $2p-1 \,\ge\, 0$, these equalities establish \textit{$p$-CR}. Reversing every ranking and vector proves the case $p \,<\, 1/2$ for $(0,1)$ and $e_p^3$. Profiles in which all voters have the same ranking rule out all remaining pairs of vectors in successive dimensions. These include a constant with a nonconstant vector, $(0,1)$ with $d_p^3$, $(1,0)$ with $e_p^3$, and, at $p \,=\, 1/2$, $d_{1/2}^{m-1}$ with $e_{1/2}^m$ or conversely.

Next let $p \,\in [1/2,1)$ and choose $q$ so that $pq$ and $(1-p)q$ are integers. In Figure~\ref{fig:obstruction}, $x$ and $z$ tie under $d_p^3$ before entry. Under $d_p^4$ after entry,
\begin{equation}\label{eq:gap}
S_z(\widetilde\succ)-S_x(\widetilde\succ) \,=\, qp(1-p) \,>\, 0.
\end{equation}
Thus $x$ ceases to be chosen. The profile in which all voters have the same ranking excludes a constant four-alternative vector. At $p \,=\, 1/2$, reversing every ranking and vector in the example also excludes $e_{1/2}^4$. For $p \,<\, 1/2$, reversing every ranking and vector in the example for $1-p$ excludes $e_p^4$. Hence only the Trivial rule remains for $p \,\in (0,1)$. At $p \,=\, 1$, equation~\eqref{eq:shape} gives the Plurality rule vectors, and entry below $x$ preserves every original first-place count. At $p \,=\, 0$, it gives the inverse-Plurality rule vectors, and entry above $x$ preserves every original last-place count.
\end{proof}

\begin{figure}[H]
\centering
\begin{minipage}[c]{0.4\textwidth}
\centering
\small Original profile\par\medskip
\profiletable{1.25cm}{3.10cm}
{\(\begin{array}{c}q\\q\end{array}\)}
{\(\begin{array}{ccc}\colorbox{blue!14}{$x$}&z&y\\z&\colorbox{blue!14}{$x$}&y\end{array}\)}
\end{minipage}
\hfill
\raisebox{0.15cm}{\tikz\draw[-{Stealth[length=2.2mm]},thick] (0,0)--(0.65,0);}
\hfill
\begin{minipage}[c]{0.49\textwidth}
\centering
\small $p$-clone extension\par\medskip
\profiletable{1.8cm}{3.65cm}
{\(\begin{array}{c}pq\\(1-p)q\\pq\\(1-p)q\end{array}\)}
{\(\begin{array}{cccc}
\colorbox{blue!14}{$x$}&\colorbox{orange!18}{$x'$}&z&y\\
\colorbox{orange!18}{$x'$}&\colorbox{blue!14}{$x$}&z&y\\
z&\colorbox{blue!14}{$x$}&\colorbox{orange!18}{$x'$}&y\\
z&\colorbox{orange!18}{$x'$}&\colorbox{blue!14}{$x$}&y
\end{array}\)}
\end{minipage}
\caption{The profile that prevents an extension to four alternatives.}
\label{fig:obstruction}
\end{figure}

\begin{proof}[Proof of Theorem~\ref{thm:plurality}]
Plurality rule satisfies all four axioms. Conversely, \textit{$1$-CI} implies \textit{$1$-CR}.
By \citet[Theorem~1]{young1975}, \textit{Anonymity}, \textit{Neutrality}, and \textit{Consistency} imply that the restriction of $F$ to each fixed alternative set is a simple or composite scoring rule.
The one-to-two argument above excludes the inverse Plurality rule for $m \,=\, 2$.
Apply \textit{$1$-CR} to the profile on $a_1,\ldots,a_k$ constructed in the preceding proof.
Before entry, every alternative ties under every component vector.
After cloning $a_1$ with $p \,=\, 1$, all original alternatives must still tie under every component vector.
Equation~\eqref{eq:shape} then shows that each such vector has the form $(a,c,\ldots,c)$.

If all component vectors are constant, the rule is Trivial. Otherwise, at the first nonconstant vector the score difference between $a_1$ and $a_1'$ is $q(a-c)$.
Since $a_1$ must remain chosen, $a \,>\, c$, so this vector is equivalent to the Plurality vector.
Every later component vector also depends only on first-place counts and therefore cannot break a tie left by that vector.
The rule on the fixed alternative set is consequently the Plurality rule or the Trivial rule.

The profile in which all voters have the same ranking also prevents using the Plurality rule for $m-1$ alternatives and the Trivial rule for $m$ alternatives, or conversely.
Thus the same rule applies for every $m$.
The Trivial rule violates \textit{$1$-CI} because it chooses $A$ before entry and $A\cup\{x'\}$ after entry.
Therefore, $F$ is the Plurality rule.
\end{proof}



\appendix
\section{Independence of the Axioms}\label{app:redundancy}
\begingroup
\renewcommand{\baselinestretch}{1.2}\normalsize

The following examples show that none of the four axioms in Theorem~\ref{thm:plurality} can be omitted.

\begin{itemize}
\item[] \textbf{Without \textit{Anonymity}.}
Fix unequal positive voter weights and choose the alternatives with the largest weighted number of first-place votes.
This rule satisfies \textit{Neutrality}, \textit{Consistency}, and \textit{$1$-Clone Invariance}.
It violates \textit{Anonymity} because exchanging two voters with different weights can change the choice set.

\item[] \textbf{Without \textit{Neutrality}.}
Fix unequal positive weights for the alternatives and multiply each alternative's first-place count by its weight.
The resulting rule satisfies \textit{Anonymity}, \textit{Consistency}, and \textit{$1$-Clone Invariance}.
It violates \textit{Neutrality} because relabelling two alternatives with different weights can change the choice set.

\item[] \textbf{Without \textit{Consistency}.}
Use the Plurality rule when the electorate has odd size.
When it has even size, choose every alternative that receives at least one first-place vote.
This rule satisfies \textit{Anonymity}, \textit{Neutrality}, and \textit{$1$-Clone Invariance}.
Consider two disjoint three-voter profiles with first-place counts $(2,1,0)$ and $(2,0,1)$ for $(a,b,c)$.
Each profile chooses $\{a\}$, whereas their union chooses $\{a,b,c\}$.
Thus, \textit{Consistency} fails.

\item[] \textbf{Without \textit{$1$-Clone Invariance}.}
Borda satisfies \textit{Anonymity}, \textit{Neutrality}, and \textit{Consistency}.
In a binary profile, let two voters rank $x$ above $y$ and three rank $y$ above $x$.
Borda chooses $\{y\}$.
After $x'$ is inserted immediately below $x$ for every voter, the scores are
\[
S_x(\widetilde\succ)=2\cdot2+3\cdot1=7
\quad\text{and}\quad
S_y(\widetilde\succ)=2\cdot0+3\cdot2=6.
\]
The choice set therefore changes, so Borda violates \textit{$1$-Clone Invariance}.
\end{itemize}
\endgroup

\section*{Declaration}

\subsection*{Conflict of interest}
The author declares that there are no conflicts of interest.

\subsection*{Acknowledgement}
I am grateful to Satoshi Nakada for his helpful comments and suggestions.

\subsection*{Disclosure on the use of AI}
The author used ChatGPT to assist with language editing and to identify possible algebraic or expository issues. All ideas and scientific content were conceived and developed by the author, who retains full responsibility for every claim made in the paper.

\begingroup
\renewcommand{\baselinestretch}{1.15}
\endgroup

\end{document}